\documentclass[a4paper,UKenglish]{lipics}

\usepackage{microtype}

\usepackage{lineno}

\theoremstyle{definition}
\newtheorem{proposition}[theorem]{Proposition}
\newtheorem{fact}[theorem]{Fact}

\newcommand*\patchAmsMathEnvironmentForLineno[1]{%
    \expandafter\let\csname old#1\expandafter\endcsname\csname #1\endcsname
    \expandafter\let\csname oldend#1\expandafter\endcsname\csname end#1\endcsname
    \renewenvironment{#1}%
    {\linenomath\csname old#1\endcsname}%
{\csname oldend#1\endcsname\endlinenomath}}%
\newcommand*\patchBothAmsMathEnvironmentsForLineno[1]{%
    \patchAmsMathEnvironmentForLineno{#1}%
\patchAmsMathEnvironmentForLineno{#1*}}%
\AtBeginDocument{%
    \patchBothAmsMathEnvironmentsForLineno{equation}%
    \patchBothAmsMathEnvironmentsForLineno{align}%
    \patchBothAmsMathEnvironmentsForLineno{flalign}%
    \patchBothAmsMathEnvironmentsForLineno{alignat}%
    \patchBothAmsMathEnvironmentsForLineno{gather}%
    \patchBothAmsMathEnvironmentsForLineno{multline}%
}

\newcommand{\nop}[1]{}

\title{Regular Language Distance and Entropy}

\begin{document}

\author[1]{Austin J. Parker}
\author[2]{Kelly B. Yancey}
\author[3]{Matthew P. Yancey}

\affil[1]{\small{Institute for Defense Analyses - Center for Computing Sciences, Bowie MD, USA; ajpark2@super.org}}
\affil[2]{\small{University of Maryland, College Park MD, USA; Institute for Defense Analyses - Center for Computing Sciences, Bowie MD, USA; kbyancey1@gmail.com}}
\affil[3]{\small{Institute for Defense Analyses - Center for Computing Sciences, Bowie MD, USA; mpyancey1@gmail.com}}
 

\maketitle

\begin{abstract}
    This paper addresses the problem of determining the distance between two
    regular languages. It will show how to expand Jaccard distance, which
    works on finite sets, to potentially-infinite regular languages.

    The entropy of a regular language plays a large role in the extension.
    Much of the paper is spent investigating 
    the entropy of a regular language.  This includes addressing
    issues that have required previous authors to rely on the upper limit of
    Shannon's traditional formulation of entropy, because its limit
    does not always exist
    \cite{ceccherini2003entropy, cui2013similarity,kuich1970entropy}.
    The paper also includes proposing a new limit based formulation for the
    entropy of a regular language and proves that formulation to both exist and
    be equivalent to Shannon's original formulation (when it exists).
    Additionally, the proposed formulation is shown to equal an analogous but
    formally quite different notion of topological entropy from Symbolic
    Dynamics -- consequently also showing Shannon's original formulation to be
    equivalent to topological entropy.

    Surprisingly, the natural Jaccard-like entropy distance is
    trivial in most cases. Instead, the {\it entropy sum} distance metric
    is suggested, and shown to be granular in certain situations.

\end{abstract}

\section{Introduction}
\label{sec:intro}

In this paper we study distances between regular expressions.  There are many motivations for this analysis. Activities
in bioinformatics, copy-detection \cite{cui2013similarity}, and network defense sometimes require large
numbers of regular expressions be managed, and metrics aid in indexing and
management of those regular expressions \cite{chan2003re}.  Further,
understanding the distance between regular languages requires an investigation
of the structure of regular languages that we hope eliminates the need for
similar theoretical investigations in the future.

A natural definition of the distance between regular languages
$L_1$ and $L_2$ containing strings of symbols from $\Sigma$
is: $\lim_{n\rightarrow\infty} \frac{\left|L_1\cap L_2\cap \Sigma^n\right|}{\left|(L_1\cup L_2)\cap \Sigma^n\right|}$.
However, this definition has a fundamental flaw: the limit does not always
exist.  Consider the distance between $(aa)^*$ and $a^*$.
When $n$ is even, the fraction is $1$, while when $n$ is odd the fraction
is $0$. Thus, the limit given above is not well defined for those two languages.

This paper addresses that flaw and examines the question of entropy and
distance between regular languages in a more general way. A fundamental
contribution will be a limit-based distance related to the above that (1)
exists, (2) can be computed from the Deterministic Finite Automata for the
associated regular language, and (3) does not invalidate expectations about the
distance between languages.

The core idea is two-fold: (1) to rely on the number of strings
{\it up-to} a given length rather than strings {\it of} a given length and (2) to use Ces\'aro averages to smooth out the behavior of the limit.  These ideas led us to develop the Ces\'aro Jaccard distance, which is proved to be well-defined in Theorem \ref{J_C_defined}.

Tied up in this discussion will be the {\it entropy} of a regular
language, which is again a concept whose common definition needs tweaking due
to limit-related considerations.

This paper is structured as follows. In Section~\ref{sec:definitions} we define
terms that will be used in the paper. Of particular importance is
Table~\ref{tab:distances}, which includes all of the distance functions defined
in this paper.  As the Jaccard distance is a natural entry point into distances
between sets, Section~\ref{sec:jaccard} will discuss the classical Jaccard distance and how best
to extend it to infinite sets. Section~\ref{sec:entropy} will discuss notions
of regular language entropy, introducing a new formulation and proving it
correct from both a channel capacity and a topological entropy point of view.
Section~\ref{sec:entropydistance} will introduce some distances based on
entropy, and show that some of them behave well, while others do not. Finally,
Section~\ref{sec:conclusion} provides a conclusion and details some potential
future work.


\section{Background}\label{sec:background}


\subsection{Related Work}
\label{sec:related}

Chomsky and Miller's seminal paper on regular languages
\cite{chomsky1958finite} does not address distances between regular languages.
It uses Shannon's notion of channel capacity (equation 7 from \cite{chomsky1958finite}) for
the entropy of a regular language: $h(L) = \lim_{\lambda\rightarrow\infty}
            \frac {\log \left|L\cap \Sigma^\lambda\right|} {\lambda} .$

While Shannon says: ``the limit in question will exist as a finite number in
most cases of interest'' \cite{shannon48mathematical},
that limit does not always exist for regular languages (consider
$(\Sigma^2)^*$).
That fact can be seen as motivating some of the analysis in this paper. Chomsky
and Miller also examine the number of sentences {\it up to} a given length,
foreshadowing some other results in this paper.

Several works since Chomsky and Miller have used this same {\it of length exactly $n$}
formula to define the entropy of a regular language
\cite{cui2013similarity,kuich1970entropy,ceccherini2003entropy}.
These works define entropy as Chomsky and Miller, but add the caveat
that they use the upper limit when the limit does not exist.
There is even a paper on non-regular languages that uses the same
entropy definition \cite{kuich1970entropy}.
Here we provide foundation for those works by showing the upper limit to
be correct (Theorem~\ref{language_entropy}). 
Further, this paper suggests an equivalent expression for entropy that may be
considered more elegant: it is a limit that exists
as a finite number for all regular languages.

There is work examining distances between {\it unary} regular languages, or
regular languages on the single character alphabet ($|\Sigma|=1$)
\cite{dassow2009similarity}.  It introduces a definition for 
Jaccard distance that will appear in this paper: $1 - \lim_{t\rightarrow\infty}\frac{\left|L_1 \cap L_2 
                                \cap\left(\bigcup_{i=0}^t\Sigma^i\right)\right|}
               {\left|\left(L_1\cup L_2\right)\cap \left(\bigcup_{i=0}^t\Sigma^i\right)\right|}.$
Further, it gives a closed form for calculating that distance between two
unary regular languages. Apart from the focus on unary regular languages,
this paper differs from \cite{dassow2009similarity} in its analysis of
the distance functions presented.  In particular, you can conclude (as a consequence of Theorem \ref{limits of cesaro}) that the above equation is mostly trivial when applied to nonunary languages -- it returns $0$ or
$1$ ``most'' of the time.

More recently, Cui {\it et al} directly address distances between regular
languages using a generalization of Jaccard distance \cite{cui2013similarity}.
That paper usefully expands the concept of Jaccard distance to regular
languages by (1) using entropy to handle infinite sized regular languages
(they use the upper limit notion of entropy described above), and
(2) allowing operations other than intersection to be used in the numerator.
Further, Cui {\it et al} suggest and prove properties of several
specific distance functions between regular languages. The distance functions
in this paper do not generalize the Jaccard distance in the same way, but
are proven to be metrics or pseudo-metrics.

Ceccherini-Silberstein {\it et al} investigate the entropy of specific
kinds of subsets of regular languages \cite{ceccherini2003entropy}.
They present a novel proof of a known fact from Symbolic Dynamics.
They use the same upper limit notion of entropy as above.
Other entropy formulations include the number of prefixes of a regular language \cite{chang2014}, but this has only been proven equivalent to entropy under restricted circumstances.

Symbolic dynamics \cite{lind1995symbolic} studies,
among other things, an object called a sofic shift. Sofic shifts are
analogous to deterministic finite automata and their shift spaces are related to regular
languages. The formulation of entropy used in this field does not suffer from
issues of potential non-existence.
This paper includes a proof that the {\it topological entropy} of a sofic
shift is equivalent to language-centric formulations in this paper: see
Theorem~\ref{language_entropy}.

Other related results from symbolic dynamics include an investigation into the
computability of a sofic shift's entropy \cite{simonsen2006computability}
and a
discussion of the lack of relationship between entropy and complexity
\cite{li1991relationship}.

There is another proposal for the topological entropy of formal
(including regular) languages that does not agree with the notions
provided in this paper \cite{schneider2015topological}. That paper develops
interesting formal results surrounding their definition of entropy,
including showing it to be zero for all regular languages.  Because their
entropy is zero for all regular languages, it will not be helpful as a distance
function for regular languages.

Several regular language distance and similarity functions are suggested in
\cite{chan2003re}. That paper constructs a natural R-tree-like index of regular
expressions. The index allows for faster matching of a string against a large
number of regular expressions. To construct the index, several distance
functions are considered. These include a max-count measure, which considers
the number of strings in both languages with length less than some constant to
be the languages' similarity; a rate-of-growth measure, which divides the sum
of strings sized $k$ to $k+d$ in one language by the sum of strings sized $k$
to $k+d$ in another language; and a minimum description length measure, which
computes the number of bits needed to encode a path through an NFA.


\subsection{Definitions and Notation}
\label{sec:definitions}

\begin{table}
    \centering
    \begin{tabular}{|l|l|l|}
        \hline
        $J'_n(L_1, L_2)$ & $n$ Jaccard Distance &
            $\frac {|W_n(L_1\triangle L_2)|}{|W_n(L_1\cup L_2)|}$ \\ 
        \hline
        $J_n(L_1, L_2)$ & $n_\leq$ Jaccard Distance &
            $\frac {|W_{\leq n}(L_1\triangle L_2)|}
                   {|W_{\leq n}(L_1\cup L_2)|}$ \\ 
        \hline
        $J_C(L_1, L_2)$ & Ces\`aro Jaccard &
            $\lim_{n\rightarrow\infty}\frac 1 n \sum_{i=1}^n J_i(L_1, L_2)$ \\ 
        \hline
        $H(L_1, L_2)$ & Entropy Distance &
            $\frac {h(L_1\triangle L_2)} {h(L_1\cup L_2)}$ \\ 
        \hline
        $H_S(L_1, L_2)$ & Entropy Sum Distance &
            $h(L_1\cap\overline{L_2}) + h(\overline{L_1}\cap L_2)$ \\ 
        \hline
    \end{tabular}
    \caption{\label{tab:distances}
        The distance functions considered in this paper are listed in this table.
    }
\end{table}

In this paper $\Sigma$ will denote a set of \textit{symbols} or the alphabet. \textit{Strings} are concatenations of these symbols, with
$\epsilon$ representing the empty string. All $\log$ operations in this paper
will be with taken base 2.  Raising a string to the power $n$
will represent the string resulting from $n$
concatenations of the original string.  A similar notion applies to sets.  In this notation, $\Sigma^5$ represents all strings of 
length $5$ composed of symbols from $\Sigma$. The Kleene star, $*$, when
applied to a string (or a set) will represent the set containing strings resulting
from any number of concatenations of that string (or of strings in that set),
including the empty concatenation. Thus, $\Sigma^*$ represents all possible
strings comprised of symbols in $\Sigma$, including the empty string.

A \textit{regular language} is a set $L\subset\Sigma^*$ which can be
represented by a \textit{Deterministic Finite Automata}, DFA for short.  A \textit{DFA} is a 5-tuple
$(Q, \Sigma, \delta, q_0, F)$, where $Q$ is a set of states, $\Sigma$ is
the set of
symbols, $\delta$ is a function from $Q\times \Sigma$ to $Q$, $q_0\in Q$ is the
\textit{initial state} and $F\subset Q$ is a set of \textit{final states}.
A regular language can also be constructed by recursive applications of
concatenation (denoted by placing regular expressions adjacent to one another),
disjunction (denoted $|$), and Kleene star (denoted $*$), to strings and
the empty string. That this construction and the DFA are equivalent is
well known \cite{Hopcroft_Ullman}.

The DFA $(Q, \Sigma, \delta, q_0, F)$ can be thought of as a directed graph whose
vertices are $Q$ with edges from $q$ to $q'$ iff there is an $s\in \Sigma$ such that
$q'=\delta(q, s)$. The transition function $\delta$ provides a labeling of the graph, that is each edge
$(q, q')$ is labeled  by the symbol $s$ when $\delta(q, s) = q'$. The adjacency matrix $A$
for a DFA is the adjacency matrix for the corresponding graph.  Thus, entries in $A$ are given by $a_{q,q'}$, where $a_{q,q'}$ is the number of edges from vertex $q$ to vertex $q'$.

For a regular language $L$, let $W_n(L)$ denote the set of words in $L$ of
length exactly $n$, i.e. $W_n(L)= L\cap \Sigma^n$, and 
let $W_{\leq n}(L)$ denote the set of words in $L$ of length at most $n$, i.e. $W_{\leq n}(L)=L\cap(\bigcup_{i=0}^n\Sigma^i)$.

Finally, we will discuss when certain distance functions are metrics.  A \textit{metric} on the space $X$ is a function $d:X\times X\rightarrow \mathbb{R}$ that satisfies
\begin{enumerate}
\item $d(x,y)\geq 0$ with equality if and only if $x=y$ for all $x,y\in X$
\item $d(x,y)=d(y,x)$ for all $x,y\in X$
\item $d(x,z)\leq d(x,y)+d(y,z)$ for all $x,y,z\in X.$
\end{enumerate}
An \textit{ultra-metric} is a stronger version of a metric, with the triangle inequality (the third condition above) replaced with the ultra-metric inequality: $d(x,z)\leq \max\{d(x,y),d(y,z)\}$ 
for all $x,y,z\in X$.
Also, there exists a weaker version, called a \textit{pseudo-metric}, which allows $d(x,y) = 0$ when $x \neq y$.

\section{Jaccard Distances}
\label{sec:jaccard}

The Jaccard distance is a well-known distance function between finite sets.  For finite sets $A$ and $B$, the \textit{Jaccard distance} between them is given by $\frac{\left|A\triangle B\right|}{\left|A\cup B\right|}=1-\frac{\left| A\cap B\right|}{\left| A\cup B \right|}$ where $A\triangle B$ represents the symmetric difference between the two sets (if $A\cup B=\emptyset$ then the Jaccard distance is $0$).  This classical Jaccard distance is not defined for infinite sets and as such, is not a suitable distance function for infinite regular languages and will need to be modified.


\subsection{Jaccard Distances using $W_{n}$ and $W_{\leq n}$}\label{finite_Jaccard}

A natural method for applying Jaccard distance to regular languages is to fix $n$ and define it as follows:

\begin{definition}[$n$ Jaccard Distance]

Suppose $L_1$ and $L_2$ are regular languages.  Define the \textit{$n$ Jaccard distance} by $J_n'(L_1, L_2) = \frac {\left|W_n(L_1\triangle L_2)\right|} {\left|W_n(L_1\cup L_2)\right|}.$
If $|W_n(L_1\cup L_2)|=0$, then $J_n'(L_1, L_2) = 0$.
\end{definition}
For fixed $n$, the above is a pseudo-metric since it is simply
the Jaccard distance among sets containing only length $n$ strings. The
following proposition points out one deficiency of $J_n'$.

\begin{proposition} \label{prop:n_jaccard_ex}
There exists a set $S = \{L_1, L_2, L_3\}$ of unary regular languages with $L_2,L_3 \subset L_1$ such that for all $n$ there exists an $i \neq j$ such that $J_n'(L_i,L_j) = 0$.
\end{proposition}

\begin{proof}
Let $L_1 = a^*$, $L_2 = (aa)^*$, and $L_3 = a(aa)^*$.  Fix $n\in\mathbb{N}$.
If $n$ is even, then $J_n'(L_1, L_2 ) = 0$, and if $n$ is odd, then $J_n'( L_1, L_3 ) = 0$.
\end{proof}

\newpage

One may also use $W_{\leq n}$ in the definition of a distance function.

\begin{definition}[$n_{\leq}$ Jaccard Distance]
For regular languages $L_1$ and $L_2$, define the \textit{$n_{\leq}$ Jaccard distance} by $J_n(L_1,L_2)=\frac{\left|W_{\leq n}(L_1\triangle L_2)\right|}{\left|W_{\leq n}(L_1\cup L_2)\right|}.$
If $|W_{\leq n}(L_1\cup L_2)|=0$, then $J_n(L_1, L_2) = 0$.
\end{definition}

The issue with $J_n'$ pointed out by Proposition~\ref{prop:n_jaccard_ex} can be proven to not be a problem for
$J_n$: see the first point of Theorem~\ref{jaccard pseudo-metric}.

\begin{theorem} \label{jaccard pseudo-metric}
The function $J_n$ defined above is a pseudo-metric and satisfies the following:
\begin{enumerate}
	\item Let $S = \{L_1, \ldots, L_k\}$ be a set of regular languages.  There exists an $n$ such that $J_n$ is a metric over $S$.  Moreover, we may choose $n$ such that 
		$n \leq \max_{i,j} (s(L_i) + 1)(s(L_j) + 1) - 1$ where $s(L_i)$ represents the number of states in the minimal DFA corresponding to $L_i$.
	\item For any fixed $n$ there exists regular languages $L, L'$ with $L\neq L'$ such that $J_n(L, L') = 0$.
\end{enumerate}
\end{theorem}

\begin{proof}
The fact that $J_n$ is a pseudo-metric follows from the fact that the standard Jaccard distance for finite sets is a metric.

(1) Let $S= \{L_1, \ldots, L_k\}$ be a fixed finite set of regular languages. For each $i \neq j$ there exists an $n_{i,j}$ such that $\left|W_{n_{i,j}}(L_i \triangle L_j)\right| \neq 0$ since $L_i \neq L_j$ and only one $L_i$ can
be $\emptyset$.
Let $n = \max_{i,j} n_{i,j}$.  Then $J_n$ is a metric over $S$.
Every regular language $L_i$ contains a word whose length is at most $s(L_i)$.
Now, we simply observe that $s(L_i \triangle L_j) \leq (s(L_i) + 1)(s(L_j) + 1) - 1$.

(2) Let $n$ be an arbitrary number and let $\Sigma' = \Sigma\cup \{z\}$,
where $z\notin \Sigma$. Take an arbitrary regular language $L$ over $\Sigma$.
Construct a regular language $L' = L \cup \{z^{n + 1}\}$ over $\Sigma'$. $L'$ is the language $L$ with the addition of the element $z^{n+1}$. When $L$ is
considered over alphabet $\Sigma'$, we have: $J_n(L, L') = 0$.
\end{proof}

For any pseudo-metric, the relation $d(x,y) = 0$ is an equivalence relation.
Thus, if we mod out by this equivalence relation, the pseudo-metric becomes a metric.

Due to the fact that one must choose a fixed $n$,
$J_n$ and $J_n'$ cannot account for the infinite nature of
regular languages. Limits based on $J_n$ and $J_n'$ are a natural next step.  However, the natural limits involving $J_n'$ and $J_n$ do not always exist.
An example showing this was given for $J_n'$ in the beginning of the 
introduction (Section~\ref{sec:intro}). A similar example applies to $J_n$.
Consider the languages given by $L_1=(a|b)^*$ and $L_2=((a|b)^2)^*$
($\Sigma=\{a,b\}$).
For these languages, $\lim_{n \rightarrow \infty} J_{2n}(L_1, L_2) = 2/3$ and $\lim_{n \rightarrow \infty} J_{2n+1}(L_1, L_2) = 1/3$.  Hence, $\lim_{n\rightarrow\infty} J_n(L_1,L_2)$ does not exist.

\begin{figure}
    \centering
    \begin{minipage}{0.45\linewidth}
        \includegraphics[width=\linewidth]{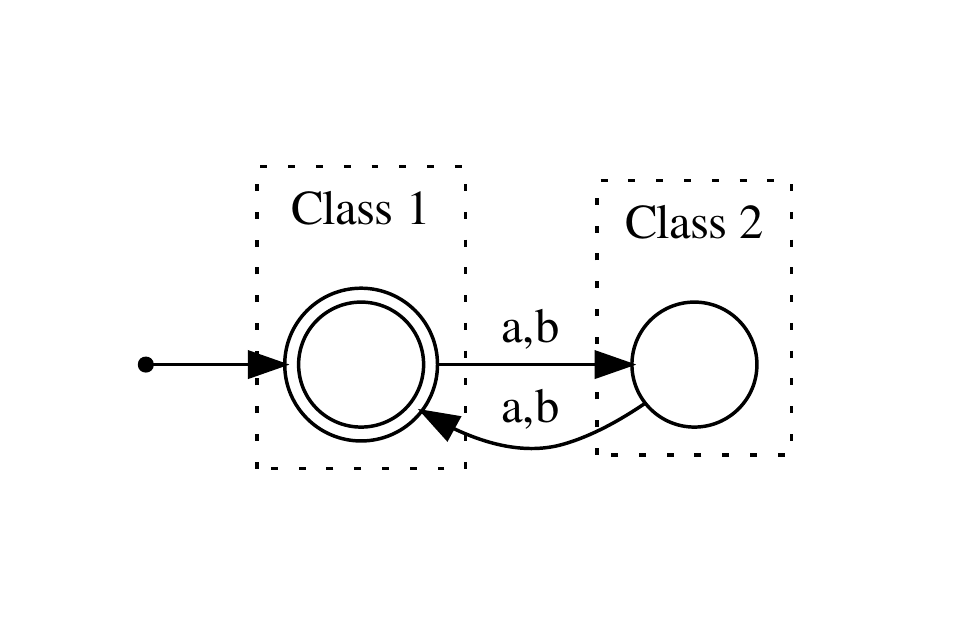}
    \end{minipage}
    \begin{minipage}{0.45\linewidth}
        \[
      \left(\begin{array}{cc} 0 & 2 \\ 2 & 0 \end{array}\right)^i
        = \left\{
            \begin{array}{cc}
            \left(\begin{array}{cc} 0 & 2^i \\ 2^i & 0 \end{array}\right)
                                        & \textrm{if i odd}\\
            \left(\begin{array}{cc} 2^i & 0 \\ 0 & 2^i \end{array}\right)
                                        & \textrm{if i even}
            \end{array}
           \right. 
       \]
    \end{minipage}\\
    \caption{
        \label{fig:period2}
        The DFA for the period 2 language $((a|b)^2)^*$ and the
        associated adjacency matrix raised to the $i^{\rm th}$ power.
    }
\end{figure}

The next theorem gives conditions for when the limit of $J_n'$ exists as $n$
goes to infinity.  Before the theorem is stated we will need some more
terminology.  Suppose $L$ is a regular language and $M$ is the corresponding
DFA.  This DFA is a labeled directed graph.  An \textit{irreducible component}
of $M$ is a strongly connected component of the graph.  That is, an irreducible
component is composed of a set of vertices such that for any pair, there is a
directed path between them.  Given an irreducible graph (or associated adjacency
matrix), we can define the period of the graph (or matrix).  The graph has
\textit{period} $p$ if the vertices can be grouped into classes that move
cyclically together with period $p$.  If $p=1$, the graph (or matrix) will be
called \textit{aperiodic}. See Figure~\ref{fig:period2} for an example of a regular
language whose DFA has period $2$.
Formally, an irreducible graph is {\it aperiodic} if
there is an $n$ such that all entries of the adjacency matrix $A$ raised to the
$n$-th power are positive \cite{marklov2015dynamical}. This is born out in
Figure~\ref{fig:period2}, all powers of that matrix contain at least one zero.
Note that matrices that
are irreducible and aperiodic are called \textit{primitive}.

\begin{theorem}\label{convergence_Jaccard}
Suppose $L_1$ and $L_2$ are regular languages.  If each irreducible component of the DFA associated to $L_1\triangle L_2$ and $L_1\cup L_2$ are aperiodic, then $\lim_{n\rightarrow \infty} J_n'(L_1,L_2)$ converges.
\end{theorem}

To build intuition for this theorem, as well as how we will frame the question of convergence in the next subsection, we will first discuss Theorem \ref{convergence_Jaccard} in the case where the DFA associated to $L_1\triangle L_2$ and $L_1\cup L_2$ are irreducible and aperiodic, i.e. primitive.  Suppose $A$ and $B$ are the adjacency matrices for $L_1\triangle L_2$ and $L_1\cup L_2$ respectively.  Perron-Frobenius theory tells us that the eigenvalue of largest modulus of a primitive matrix is real and unique.  Let $(v,\lambda)$ and $(x,\rho)$ be eigenpairs composed of the top eigenvalues for $A$ and $B$ respectively.  Notice that $i_A A^n f_A$, where $i_A$ is the vector whose $j$th entry is $1$ if $j$ is an initial state in $A$ and $0$ otherwise (a similar definition for final states defining $f_A$ holds), represents words in $L_1\triangle L_2$ of length $n$.  If we write $f_A=c_1v+c_2w$ and $f_B=d_1x+d_2y$, then $i_A A^n f_A$ converges to $\lambda^nc_1i_A v$, and $i_B B^n f_B$ converges to $\rho^nd_1i_B x$ as $n$ goes to infinity.  This convergence is guaranteed because $\lambda$ and $\rho$ are unique top eigenvalues.  Thus, 
\[\lim_{n\rightarrow\infty} J_n'(L_1,L_2) = \lim_{n\rightarrow\infty} \left(\frac{\lambda}{\rho}\right)^n \frac{c_1i_Av}{d_1i_Bx}\]
and the limit converges ($\lambda\leq\rho$ because $L_1\triangle L_2\subseteq L_1\cup L_2$). 

The general case  of Theorem \ref{convergence_Jaccard},
which does not assume $L_1\triangle L_2$ and $L_1\cup L_2$ have irreducible
matrices, is more complicated.
To give a quick overview:
recall that $|W_n(L_1 \triangle L_2)|$ and $|W_n(L_1 \cup L_2)|$ can be calculated using powers of specific matrices.
An understanding of the asymptotic behavior of $A^n$ for large $n$ was finally beginning to be developed several decades after Chomsky and Miller defined regular languages.
In 1981 Rothblum \cite{Rothblum_expansion_of_sums} proved that for each non-negative matrix $A$ with largest eigenvalue $\lambda$, there exists $q \geq 1$ (called the \emph{period} of $A$) and polynomials $S_0(x), S_1(x), \ldots, S_{q-1}(x)$ (whose domain is the set of real numbers and whose coefficients are matrices) such that for all whole numbers $0 \leq k \leq q-1$ we have that $\lim_{n \rightarrow \infty} \left(A/\lambda\right)^{qn+k} - S_k(qn+k) = 0$. Since $q=1$ in the case we are interested in (i.e. Theorem \ref{convergence_Jaccard}), $\lim_{n\rightarrow \infty} J_n'(L_1,L_2)$ converges.


\subsection{Ces\`aro Jaccard}\label{C_Jaccard}

For a sequence of numbers $a_1, a_2, \ldots$, a Ces\`{a}ro summation is $\lim_{n \rightarrow \infty} \frac{1}{n} \sum_{i=1}^n a_i$ when the limit exists.
The intuition behind a Ces\`{a}ro summation is that it may give the ``average value'' of the limit of the sequence, even when the sequence does not converge.
For example, the sequence $a_j = e^{\alpha i j}$ has Ces\`{a}ro summation $0$ for all real numbers $\alpha \neq 0$.  This follows from the fact that rotations of the circle are uniquely ergodic \cite{Hasselblatt_Katok}.
Not all sequences have a Ces\`{a}ro summation, even when we restrict our attention to sequences whose values lie in $[0,1]$.
For example, the sequence $b_i$, where $b_i = 1$ when $2^{2n} < i < 2^{2n+1}$ for some $n \in \mathbb{N}$ and $b_i = 0$ otherwise has no Ces\`{a}ro summation.
However, we will be able to show that the Ces\`aro average of Jaccard distances
does exist.

To that end, another limit based distance is the Ces\`aro average of the $J_n$
or $J_n'$. 

\begin{definition}[Ces\`aro Jaccard Distance]
Suppose $L_1$ and $L_2$ are regular languages.  Define the \textit{Ces\`aro Jaccard distance} by $J_C(L_1,L_2)=\lim_{n\rightarrow\infty} \frac{1}{n}\sum_{i=1}^n J_i(L_1,L_2).$
\end{definition}

\begin{fact}
The Ces\`aro Jaccard distance inherits the pseudo-metric property from $J_n$.
\end{fact}

The Ces\`aro Jaccard distance is theoretically better than the
above suggestions in Section \ref{finite_Jaccard} since it can be shown to exist for all regular languages.
\begin{theorem}\label{J_C_defined}
Let $L_1$ and $L_2$ be two regular languages.  Then, $J_C(L_1,L_2)$ is well-defined.  That is, $\lim_{n\rightarrow\infty} \frac{1}{n}\sum_{i=1}^n J_i(L_1,L_2)$
exists.
\end{theorem}

Taking a Ces\`aro average also has the ability to smooth out $J_n'$: it exists when $J_n'$ is used in place of $J_n$.
However, the following example motivates our use of $J_n$ in the definition of Ces\`aro Jaccard. 

\begin{example}\label{J_C_example}
Let $L_1=((a|b|c)^2)^*|(d|e)^*$ and $L_2=((a|b|c)^2)^*|(f|g)^*$.
We have that 
$\lim_{n \rightarrow \infty} J_{2n}(L_1, L_2) = 0$
and
$\lim_{n \rightarrow \infty} J_{2n}'(L_1, L_2) = 0$.
Furthermore,
$\lim_{n \rightarrow \infty} J_{2n+1}(L_1, L_2) = 0$ and
$\lim_{n \rightarrow \infty} J_{2n+1}'(L_1, L_2) = 1$.
Thus, $J_C(L_1, L_2) = 0$, but it would equal
$1/2$ if $J_n$ was replaced with
$J_n'$. Because the overwhelming majority of $L_1$ and $L_2$ are
in $L_1\cap L_2$, $0$ is a better value for $J_C(L_1, L_2)$ than $\frac 1 2$,
so we prefer $J_n$ over $J_n'$ in the definition. 
\end{example}

A full proof of Theorem \ref{J_C_defined} will wait until Section \ref{proof of J_C}, as it uses ideas from Section \ref{sec:entropy}.
We present here a result that is the core of the proof (the case when $\lambda = 1$ is similar to a result in \cite{Rothblum_expansion_of_sums}).

\begin{theorem} \label{sum of matrices}
(Part I) Let $A$ be the adjacency matrix for a DFA representing a regular language $L$, and let $\lambda$ be the largest eigenvalue of $A$.
Let $q$ and $S_0(x), S_1(x), \ldots, S_{q-1}(x)$ be as in Rothblum's theorem (indices will be taken modulo $q$).
Let $d$ be the largest degree of the polynomials $S_0(x), S_1(x), \ldots, S_{q-1}(x)$, and let $s_\ell$ be the coefficient of $x^d$ in $S_\ell(x)$.
\begin{enumerate}
	\item If $\lambda < 1$, then $L$ is finite.
	\item If $\lambda = 1$, then $\lim_{n \rightarrow \infty} \frac{1}{n^{d+1}} \sum_{i=1}^n A^i = \frac1{q(d+1)}\sum_{i=0}^{q-1}s_\ell$.
	\item If $\lambda > 1$, then $\lim_{n \rightarrow \infty} \frac{1}{(qn+k)^{d}}\lambda^{-(qn+k)} \sum_{i=1}^{qn+k} A^i = \frac1{\lambda^k(1-\lambda^{-q})}\sum_{\ell = k-q+1}^k \lambda^\ell s_\ell$.
\end{enumerate}

(Part II) Let $u$ be a row vector, $v$ a column vector, and let $S_j'(x) = u S_j(x) v$, so that $S_j'(x)$ is a polynomial with real numbers as coefficients.
Let $d'$ be the largest degree of the polynomials $S_0'(x), S_1'(x), \ldots, S_{q-1}'(x)$, and let $s_\ell'$ be the coefficient of $x^{d'}$ in $S_\ell'(x)$.
\begin{enumerate}
	\item If $\lambda = 1$, then $\lim_{n \rightarrow \infty} \frac{1}{n^{d'+1}} \sum_{i=1}^n u A^i v = \frac1{q(d'+1)}\sum_{i=0}^{q-1}s_\ell'$.
	\item If $\lambda > 1$, then $\lim_{n \rightarrow \infty} \frac{1}{(qn+k)^{d'}}\lambda^{-(qn+k)} \sum_{i=1}^{qn+k} u A^i v = \frac1{\lambda^k(1-\lambda^{-q})}\sum_{\ell = k-q+1}^k \lambda^\ell s_\ell '$.
\end{enumerate}

\end{theorem}

\begin{proof} 
We omit the proof of Part II, as it will be clear from the proof of Part I.

See \cite{Rothblum_chapter} for a background on linear algebra.
The largest eigenvalue of a non-negative matrix is at least the value of the smallest sum of the entries in a row of an irreducible component.
Because the adjacency matrix for a DFA has integer entries, this implies that either $\lambda = 0$ or $\lambda \geq 1$.
If $\lambda = 0$, then $A$ is nilpotent (in other words, there exists an $n'$ such that $A^{n'} = 0$), which means $L$ is finite.
So assume $\lambda \geq 1$.

First consider the case when $\lambda = 1$.
First, we break the sum up as $\sum_{i=1}^{qn+k} A^i = \sum_{i = k-q+1}^k \sum_{j=1}^n A^{qj + i} + O(1)$.
Each part $\sum_{j=1}^n A^{qj + i} = \sum_{j=1}^n (o(1) + S_i(qj + i))$ will be approximated as $\int_{1}^n S_i(qx + i) dx$, which will be sufficiently accurate because $\sum_{x=1}^n x^b - \int_{1}^n x^bdx \leq O\left(n^b\right)$.
The result follows from $\lim_{n \rightarrow \infty} (qn + i)^{-(d+1)} \int_i^n S_i (qx + i) dx = \frac{s_\ell}{q(d+1)}$. 

Finally, suppose that $\lambda > 1$. 
Let $\epsilon > 0$ be an arbitrary number.
Let $N = qn + k$, $N_* \approx q(n - \log^2(n))$ such that $N_* \equiv k(mod\ q)$, and $n_* = (N_* - k)/q$.
For a matrix $M$, let $\|M\|_e$ denote the maximum magnitude among the entries of $M$.
Notice that the following terms converge to zero:
\begin{itemize}
	\item $\|\lambda^{-N}\sum_{i=1}^{N_*}A^i\|_e \leq O\left(n\lambda^{-\log^2(n)}\right) \leq O\left( n^{-1} \right)$.
	\item For all $\ell$ and $N' \geq N_*$, we have that $\|S_\ell(N') - S_\ell(N)\|_e / \|S_\ell(N)\|_e \leq O\left(\frac{\log^2(n)}{n}\right)$.
	\item For all $\ell$ and $n' \geq n_*$, Rothblum's theorem states that $\|\left(A/\lambda\right)^{qn'+\ell} - S_\ell(n')\|_e$ converges to $0$ exponentially.
\end{itemize}
Let $\delta > 0$ be a number such that $\delta\left(2q\frac{1}{1 - \lambda^{-q}} + 1\right) < \epsilon$.
Let $n$ be large enough such that each of the following terms is less than $\delta$:
\begin{itemize}
	\item $\|\lambda^{-N}\sum_{i=1}^{N_*}A^i\|_e$,
	\item $N^{-d}\|S_\ell(N') - S_\ell(N)\|_e$ for all $0 \leq \ell < q$, $N' > N_*$, and 
	\item $\|\left(A/\lambda\right)^{qn'+\ell} - S_\ell(qn'+\ell)\|_e$ for all $0 \leq \ell < q$, $n' > n_*$.
\end{itemize}
By the triangle inequality, for all $0 \leq \ell < q$, $n' > n_*$,  we have that $N^{-d}\|\left(A/\lambda\right)^{qn'+\ell} - S_\ell(N)\|_e < 2\delta$.
Therefore $N^{-d}\|A^{qn'+\ell} - \lambda^{qn'+\ell}S_\ell(N)\|_e < 2\delta\lambda^{qn'+\ell}$
Adding this inequality up across all values of $0 \leq \ell < q$ and $n' > n_*$ (note that $\sum_{i = N_* + 1}^{N} A^{i} = \sum_{\ell = k-q+1}^{k} \sum_{n'=n_*}^n A^{qn'+\ell}$), we have that
$$ \lambda^{-N} N^{-d} \left\| \sum_{i = N_* + 1}^{N} A^{i} - \sum_{\ell = k-q+1}^{k} \sum_{n'=n_*}^n \lambda^{qn'+\ell} S_\ell(N) \right\| < q\sum_{n'=n_*}^{n}2 \delta\lambda^{q(n' - n)} \leq 2q\delta \frac{1}{1 - \lambda^{-q}}.$$

By the triangle inequality, we have that 
$$ \left\| \lambda^{-N} N^{-d}\sum_{i = 1}^{N} A^{i} - \sum_{n'=n_*}^n \sum_{\ell = k-q+1}^k \lambda^{qn' + \ell - N} N^{-d}S_\ell(N) \right\|$$
$$ \leq \lambda^{-N} N^{-d} \left\|  \sum_{i = N_*+1}^{N} A^{i} - \sum_{n'=n_*}^n \sum_{\ell = k-q+1}^k \lambda^{qn'+\ell} S_\ell(N) \right\| + \left\| N^{-d} \lambda^{-N} \sum_{i=1}^{N_*} A^i\right\|,$$
which is less than $\delta\left(2q\frac{1}{1 - \lambda^{-q}} + 1\right) < \epsilon$.
So if the limit of $\sum_{n'=n_*}^n \sum_{\ell = k-q+1}^k \lambda^{qn' + \ell - N} N^{-d}S_\ell(N)$ exists, then the limit of $\lambda^{-N} N^{-d}\sum_{i = 1}^{N} A^{i}$ exists and is equal.
Because as $N \rightarrow \infty$ we have that $n - n^* \approx \log^2(n) \rightarrow \infty$, we conclude that 
\begin{eqnarray*}
\frac1{\lambda^k(1-\lambda^{-q})}\sum_{\ell = k-q+1}^k \lambda^\ell s_\ell 	& = &	\sum_{\ell = k-q+1}^k \left(\lim_{N \rightarrow \infty} N^{-d} S_\ell(N) \right) \lambda^{\ell - k} \left(\lim_{m \rightarrow \infty} \sum_{i=0}^m \lambda^{-qi} \right) \\
										& = &	\lim_{N \rightarrow \infty} \sum_{\ell = k-q+1}^k N^{-d}S_\ell(N) \sum_{n'=n_*}^n \lambda^{qn' + \ell - N} \\
										& = &   \lim_{n \rightarrow \infty} \frac{1}{N^{d}}\lambda^{-N} \sum_{i=1}^{N} A^i .
\end{eqnarray*}
\end{proof}


\section{Entropy}
\label{sec:entropy}

In this section we develop the idea of topological entropy for a certain type of dynamical system and show how it relates to a quantity that we have identified as the language entropy.  Then, we will show how Ces\'aro Jaccard is related to entropy.


\subsection{Topological Entropy}

Topological entropy is a concept from dynamical systems where the space is a compact metric space and the map defined there is continuous \cite{lind1995symbolic}. In dynamics, successive applications of the map are applied and the long term behavior of the system is studied.  An orbit of a point $x$ for the map $T$ is the set $\{T^n(x) \; : \; n\in\mathbb{Z}\}$.  Topological entropy is an abstract concept meant to determine the exponential growth of distinguishable orbits of the dynamical system up to arbitrary scale.  A positive quantity for topological entropy reflects chaos in the system \cite{chaos}.
This concept was motivated by Kolmogorov and Sinai's theory of measure-theoretic entropy in ergodic theory 
\cite{kolmogorov1959entropy,sinai1959notion}, which in turn is related to Shannon entropy \cite{shannon48mathematical}.  An example of a topological dynamical system is a sofic shift, which is a symbolic system that is intricately related to DFA. Instead of defining the topological entropy of a sofic shift symbolically, which is classical, we will use the graph theoretic description.

A \textit{sofic shift} can be thought of as the space of biinfinite walks
(i.e. walks with no beginning and no end) on a right-solving labeled directed graph (a right-solving labeled graph has a unique label for each edge leaving a given node).  Suppose $G$ is a directed graph where $V$ is the set of vertices and $E$ is the set of edges of $G$.  Furthermore, suppose that every edge in $E$ is labeled with a symbol from $\Sigma$, and that there is at most one outgoing edge from each vertex with a given label (i.e. \textit{right-solving}).  Note that this construction is similar to a DFA, however there are no initial and final states.  A biinfinite walk on $G$ with a specified base vertex is an infinite walk in both directions (forward and backward) on the graph.  This biinfinite walk corresponds to a biinfinite string of symbols from $\Sigma$. See Figure~\ref{fig:sofic_ex}.

\begin{figure}
    \centering
    \begin{minipage}{0.45\linewidth}
        \includegraphics[width=\linewidth]{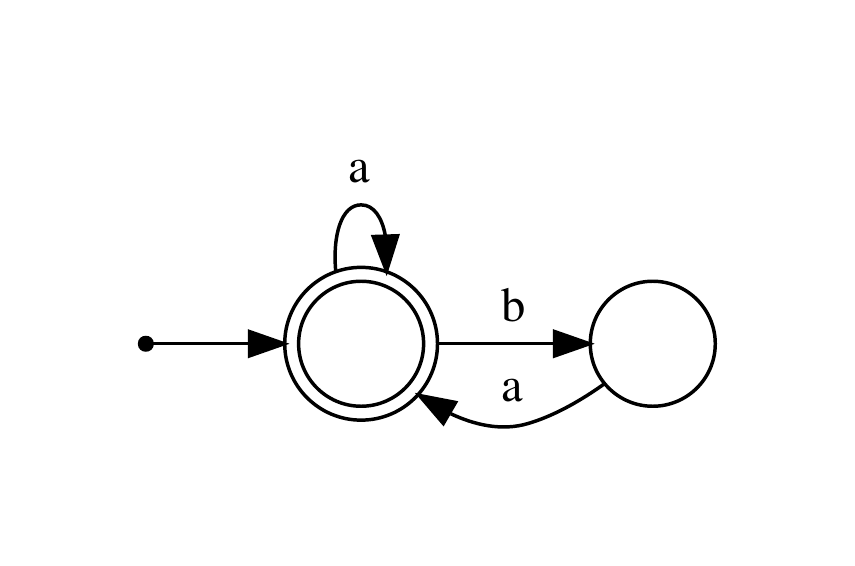}\\
        $aaa$, $aaba$, $ba$, $aaaaaaa$, $aaabaabaa$, ...
    \end{minipage}
    \hfill
    \begin{minipage}{0.45\linewidth}
        \includegraphics[width=0.805\linewidth]{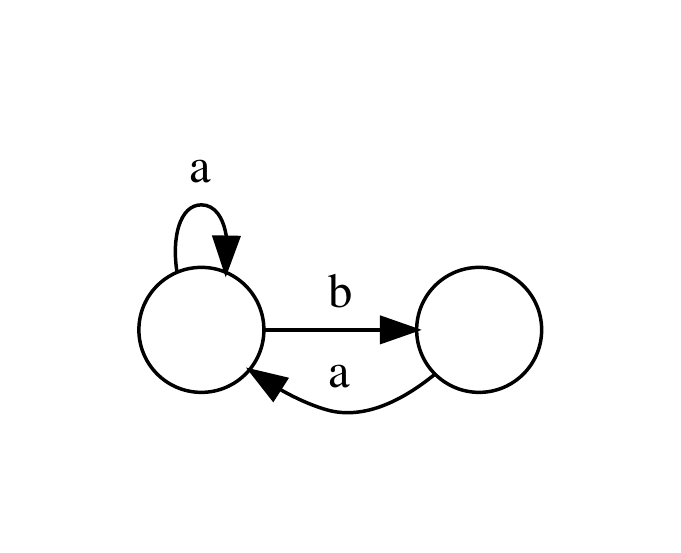}\\
        ... $aabaabaabaaaabaaaaaabaaaaab$ ...
    \end{minipage}
    \caption{
        \label{fig:sofic_ex}
        A DFA with some accepted strings and
        a sofic shift with a piece of a biinfinite string.
    }
\end{figure}

We will call a finite block of symbols \textit{admissible} if there is a biinfinite string of symbols corresponding to a biinfinite walk on $G$ and this finite block appears somewhere within the biinfinite string. 
Note that all words in the DFA's language will contain a substring of almost the same length that is an admissible block, while not all admissible blocks will be in the
associated DFA's language.
Denote the set of admissible blocks of length $n$ corresponding to $G$ by $B_n(G)$.  The \textit{topological entropy} of the sofic shift represented by the right-solving labeled graph $G$ is denoted by $h_t(G)$ and is defined by
\[h_t(G)=\lim_{n\rightarrow \infty}\frac{\log \left|B_n(G)\right|}{n}.\]

Using Perron-Forbenius theory it has been proven that the topological entropy of a sofic shift represented by a right-solving labeled graph $G$ is equal to the log base $2$ of the spectral radius of the adjacency matrix of $G$ \cite{lind1995symbolic}.  That is, the topological entropy is given by the log of the adjacency matrix's largest
modulus eigenvalue. Algorithms for computing eigenvalues are well known
and run in time polynomial in the width of the matrix \cite{francis1961qr}.

As you can see, sofic shifts are very similar to DFA.  Given a DFA, $M$, one
can construct a sofic shift by thinking of $M$ as a labeled directed graph and
creating the \textit{trim graph} by removing all states that are not part of an
accepting path. Information regarding initial and final states is no longer
needed. Note that the graph $M$ is naturally right-solving because of the determinism of DFA.
It is also easiest to remove from $M$ all vertices that do not have both an
outgoing and incoming edge (since we are now interested in biinfinite walks).
The resulting graph is called the \textit{essential graph}.  
At this point one
is free to apply the above definition and compute the topological entropy of
the sofic shift corresponding to the DFA.  This quantity can be computed by analyzing the irreducible components.

\begin{theorem}[\cite{lind1995symbolic}]\label{entropy_irr_components}
Suppose that $G$ is the labeled directed graph associated to a sofic shift.  If $G_1,\ldots, G_k$ are the irreducible components of $G$, then 
\[h_t(G)=\max_{1\leq i \leq k} h_t(G_i).\]
\end{theorem}

In the next subsection we will
introduce the language entropy and show that it is the same as the topological
entropy of the sofic shift corresponding to a DFA.


\subsection{Language Entropy}

Traditionally, the entropy of a regular language $L$ (also called the \textit{channel capacity} \cite{chomsky1958finite} or \textit{information rate} \cite{Cui_Dang_etal}) is defined as $\limsup_{n\rightarrow\infty}\frac{\log \left| W_n(L)\right|}{n}$.  This limit may not exist and so an upper limit is necessary. We will show that this upper limit is realized by the topological entropy of the corresponding sofic shift and define another notion of language entropy, which is preferable since an upper limit is not necessary.

\begin{definition}[Language Entropy]
Given a regular language $L$ define the \textit{language entropy} by 
$h(L)=\lim_{n\rightarrow\infty} \frac{\log \left| W_{\leq n}(L)\right|}{n}.$
\end{definition}

\begin{theorem}\label{language_entropy}
Let $L$ be a non-empty regular language over the set of symbols $\Sigma$, and let $G$ be the labeled directed graph of the associated sofic shift.
We have that 
\[\limsup_{n \rightarrow \infty} \frac{\log\left|W_n(L) \right|}{n} = h_t(G). \]
Moreover, for a fixed language $L$ there exists a constant $c$ such that there is an increasing sequence of integers $n_i$ satisfying $0 < n_{i+1} - n_i \leq c$ and 
\[\lim_{i \rightarrow \infty} \frac{\log\left|W_{n_i}(L) \right|}{n_i} =  h_t(G). \]
\end{theorem}

As a corollary to this theorem we obtain an important statement regarding the connection between topological entropy (from dynamical systems) and language entropy (similar to Shannon's channel capacity).

\begin{corollary}\label{language_entropy_cor}
Let $L$ be a non-empty regular language over the set of symbols $\Sigma$, and let $G$ be the labeled directed graph of the associated sofic shift.
Then,
\[h(L) = \lim_{n \rightarrow \infty} \frac{\log\left|W_{\leq n}(L) \right|}{n} =  h_t(G). \]
\end{corollary}

\noindent \textbf{Proof of Theorem \ref{language_entropy}.}  
Let $\lambda$ be the topological entropy of the sofic shift.

Let $(Q, \Sigma, \delta, q_0, F)$ be a DFA for $L$, and let $w = w_1, \ldots, w_n \in L$.
For brevity, let $n' = |Q|$.
Because we are working with some fixed language, we assume that $n' \leq O(1)$.
Recall that the Pumping Lemma states that when $n > n'$, there exists a pair $i,j$ such that $1 \leq i < j \leq n'$ and 
$$ w_1, \ldots, w_{i-1} \left(w_i, \ldots, w_j\right)^* w_{j+1} \ldots, w_n \subseteq L.$$
The proof of this statement uses the fact that if the states of $Q$ seen as $w$ streams by are $q_0, q_1, \ldots, q_n$, then there exists a pair $i,j$ as above such that $q_i = q_j$ by the pigeon hole principle.

We claim that $\{q_i, \ldots, q_j\}$ are vertices in $G$.
For $i < \ell < j$, each vertex $q_\ell$ has an incoming edge (from vertex $q_{\ell-1})$ and an outgoing edge (to vertex $q_{\ell+1}$).
Because $q_i = q_j$, this vertex also has an incoming edge (from vertex $q_{j-1})$ and an outgoing edge (to vertex $q_{i+1}$).
Therefore, this cycle is part of the essential graph.
This proves the claim.

We can iterate this procedure on the word $w_1, \ldots, w_{i-1}, w_{j+1} \ldots, w_n$ to find another subword that is admissible.
We can inductively do this until at most $n'$ characters remain.
By construction, if vertices $q_i,q_j \in G$ and $i \leq \ell \leq j$, then $q_\ell \in G$.
It follows that there exists an $i$ and a $j$ such that $q_i, \ldots, q_j$ is in $G$ and $j - i \geq n - n'$.
Therefore the number of words in $L$ of length $n$ is at most the number of admissible blocks of length $n-O(1)$ times $O(|\Sigma|^{O(1)})$ choices for the prefix $w_1, \ldots, w_{i-1}$ and the suffix $w_{j+1}, \ldots, w_n$, which implies that 
$$ \limsup_{n \rightarrow \infty} \frac{\log\left| W_n \right|}{n} \leq \lambda . $$

Let $w = w_1,w_2, \ldots, w_{n-2n'}$ be an admissible block from the sofic shift using vertices $q_1', q_2' \ldots, q_{n-2n'}'$.
By the definition of a trim graph, there exists paths in our DFA $q_0, q_1, \ldots, q_i$ and $q_j, q_{j+1}, \ldots, q_{k}$ such that $q_i = q_1'$, $q_j = q_{n-2n'}'$ and $q_k \in F$.
We may choose these paths to be minimal, which implies that no state is repeated.  Thus, $i \leq n'$ and $k-j \leq n'$.
Therefore the path $q_0, q_1, \ldots, q_i , q_2', \ldots, q_{n-2n'} , q_{j+1}, \ldots, q_k$ is a valid path in our DFA of length between $n-2n'$ and $n$, and corresponds to a word in $L$ that contains $w$ as a subword.

So each admissible block from the sofic shift of length $n-2n'$ appears in some word of $L$ whose length is between $n-2n'$ and $n$.
Each word in $L$ of length at most $n$ may contain at most $2n'$ distinct admissible blocks of length $n-2n'$ (one for each substring starting
at offsets 0, 1, 2, ..., $2n'$).
Therefore, there exists an $m$ such that $n-2n' \leq m \leq n$ and $|W_m| \geq \frac{1}{(2n')^2} \left|B_{n-2n'}(G)\right|$.
Because $n' \leq O(1)$, this proves the second part of the theorem with $c = 2n'$. It also implies that
$$ \limsup_{n \rightarrow \infty} \frac{\log\left| W_n \right|}{n}
\geq \lambda . $$ 
Which suffices to prove the first part of the theorem.  \qed

\noindent \textbf{Proof of Corollary \ref{language_entropy_cor}.}
Let $\lambda$ be the topological entropy of the sofic shift.

To show that $\lim_{n\rightarrow\infty} \frac{\log \left| W_{\leq n}(L)\right|}{n} = \lambda$ we will show that 
$$\lambda \leq \liminf_{n\rightarrow\infty}\frac{\log \left| W_{\leq n}(L)\right|}{n} \leq \limsup_{n\rightarrow\infty} \frac{\log \left| W_{\leq n}(L)\right|}{n} \leq \lambda.$$

Let $\left| W_n(L)\right| = a_n$.  For fixed $n$, let $a_{k_n}=\max\left(a_1,\ldots, a_n\right)$.  Observe that,
\begin{eqnarray*}
\limsup_{n\rightarrow\infty} \frac{\log \left| W_{\leq n}(L)\right|}{n} &=& \limsup_{n\rightarrow\infty} \frac{\log \left( a_1+\cdots + a_n\right)}{n} \\
&\leq& \limsup_{n\rightarrow\infty}\frac{\log\left( na_{k_n}\right)}{n}\\
&\leq& \limsup_{n\rightarrow\infty}\frac{\log\left( a_{k_n}\right)}{k_n} = \lambda
\end{eqnarray*}
by Theorem \ref{language_entropy}.

For the lower bound we will use the second part of Theorem \ref{language_entropy}.  For fixed $n$, let $n_k$ be the largest element from the subsequence $\left( n_i\right)$ (from Theorem \ref{language_entropy}) such that $n_k \leq n$.  In this case, $n-n_k \leq c$ where $c$ is given in the theorem.  Thus,
\begin{eqnarray*}
\liminf_{n\rightarrow\infty}  \frac{\log \left| W_{\leq n}(L)\right|}{n} &=& \liminf_{n\rightarrow\infty} \frac{\log \left( a_1+\cdots + a_n\right)}{n} \\
&\geq& \liminf_{n\rightarrow\infty}\frac{\log \left(a_{n_k}\right)}{n_k + (n-n_k)}\\
&\geq& \liminf_{n\rightarrow\infty}\frac{\log \left(a_{n_k}\right)}{n_k + c} =\lambda. \qed
\end{eqnarray*}

There are some simple properties of language entropy
which will be useful later.  The first is a simple re-phrasing of Corollary \ref{language_entropy_cor}.

\begin{lemma} \label{asymptotic growth}
For any regular language $L$, we have that $\left|W_{\leq n}(L)\right| = 2^{n(h(L) + o(1))}$.
\end{lemma}

\begin{lemma}
    \label{lemma:union_ent}
    Suppose $L_1$ and $L_2$ are regular languages over $\Sigma$.  The following hold:
    \begin{enumerate}
        \item If $L_1\subseteq L_2$, then $h(L_1)\leq h(L_2)$.
        \item $h(L_1\cup L_2) = max(h(L_1), h(L_2))$
        \item $\max(h(L_1), h(\overline{L_1})) = \log \left| \Sigma\right|$
        \item If $h(L_1)<h(L_2)$, then $h(L_2\setminus L_1)=h(L_2)$.
        \item If $L_1$ is finite, then $h(L_1) = 0$.
    \end{enumerate}
\end{lemma}

\begin{proof} Each part is proven in turn:
    
      (1) When $L_1\subseteq L_2$, $W_{\leq n}(L_1)\subseteq W_{\leq n}(L_2)$.  Thus
            $\lim_{n\rightarrow\infty}\frac {\log |W_{\leq n}(L_1)|}{n}
            \leq \lim_{n\rightarrow\infty}\frac {\log |W_{\leq n}(L_2)|}{n}$.
            
      (2)  This is a consequence of Theorem \ref{entropy_irr_components}.

      (3) Notice that $L_1\cup\overline{L_1}=\Sigma^*$ and $h(\Sigma^*)=\log \left| \Sigma\right|$.  The result follows by part 2.

      (4)  Notice that $L_2 = (L_2\setminus L_1) \cup (L_1\cap L_2)$.  Since $L_1\cap L_2 \subset L_1$ we have that $h(L_1\cap L_2)\leq h(L_1)< h(L_2)$ by part 1.  Thus 
      $h(L_2) = h((L_2\setminus L_1) \cup (L_1\cap L_2)) =\max (h(L_2\setminus L_1), h(L_1\cap L_2)) = h(L_2\setminus L_1)$.
      
      (5)  This is trivial.
             
\end{proof}


\subsection{Proof of Theorem \ref{J_C_defined}} \label{proof of J_C}

Recall Rothblum's theorem as discussed in Section \ref{finite_Jaccard} that for a matrix $A$ there exists $q$ and matrix polynomials $S_0(x), S_1(x), \ldots, S_{q-1}(x)$ such that $\lim_{n \rightarrow \infty} \left(A/\lambda\right)^{qn+k} - S_k(qn+k) = 0$ for each $0 \leq k < q$.
Let $S_k'(x) = i_A S_k(x) f_A$, so that $S_k'(x)$ is a polynomial with real numbers as coefficients.
We begin by providing a proof that $J_C'(L_1, L_2) = \frac{1}{n} \sum_{i=1}^n J_i'(L_1, L_2)$ is well-defined if $S_k'(x) \neq 0$ for all $k$.

Rothblum's result says that $|W_{qn+k}(L_1 \triangle L_2)| = \lambda^{qn+k} (S_k'(qn+k) + o(1))$.
We can apply the same argument to $|W_n(L_1 \cup L_2)|$, and let $Q$ be the least common multiple of the periods.
For each $k$, it becomes clear that $\lim_{n \rightarrow \infty} \frac{|W_{Qn+k}(L_1 \triangle L_2)|}{|W_{Qn+k}(L_1 \cup L_2)|}$ will either (1) exponentially decay (it can not exponentially grow as Jaccard distances are at most $1$) or (2) if the value of $\lambda$ is the same for $|W_n(L_1 \triangle L_2)|$ and $|W_n(L_1 \cup L_2)|$, then $\frac{|W_{Qn+k}(L_1 \triangle L_2)|}{|W_{Qn+k}(L_1 \cup L_2)|}$ becomes a ratio of polynomials whose limiting behavior is determined by the leading coefficients of $S_k'(x)$.
And so $J_C'(L_1, L_2) = \frac{1}{n} \sum_{i=1}^n J_i'(L_1, L_2)$ will be the average of these values.

It should be clear that the same result holds for $J_C(L_1, L_2)$ using Part II of Theorem \ref{sum of matrices} in replace of Rothblum's theorem if for each $k$, $\frac1{q(d'+1)}\sum_{i=0}^{q-1}s_\ell' \neq 0$ (when $\lambda = 1$) or $\frac1{\lambda^k(1-\lambda^{-q})}\sum_{\ell = k-q+1}^k \lambda^\ell s_\ell ' \neq 0$ (when $\lambda > 1$).

Now let us discuss what happens when $S_k'(x) = 0$.
This can happen; and it \emph{will} happen if $h(L) < \log(|\Sigma|)$ and the trash state was not trimmed from the DFA.

\begin{example} \label{J_C_example again}
Let us return to Example \ref{J_C_example}.
The matrix 
$$A = \left(\begin{array}{ccccc} 0 & 3 & 0 & 2 & 2 \\ 0 & 0 & 3 & 0 & 0 \\0 & 3 & 0 & 0 & 0 \\0 & 0 & 0 & 2 & 0 \\0 & 0 & 0 & 0 & 2  \end{array}\right)$$
 and vectors $i = (1, 0, 0, 0, 0)$, $f_\cup = (0, 0, 1, 1, 1)^t$, and $f_\triangle = (0,0,0,1,1)^t$ satisfy $|W_n(L_1 \cup L_2)| = i A^n f_\cup$ and $|W_n(L_1 \triangle L_2)| = i A^n f_\triangle$.
Note that $A$ represents a trim DFA for language $L_1 \cup L_2$, but the DFA is not trim for the language $L_1 \triangle L_2$.
Applying Rothblum's theorem to $A$ gives us $q = 2$, each of $S_0(x), S_1(x)$ has degree $0$, $\lambda = 3$, and 
$$S_j(x) = \frac{1}{2}\left(\begin{array}{ccccc} 0 & 1 & 1 & 0 & 0 \\ 0 & 1 & 1 & 0 & 0 \\0 & 1 & 1 & 0 & 0 \\0 & 0 & 0 & 0 & 0 \\0 & 0 & 0 & 0 & 0  \end{array}\right) + \frac{(-1)^j}{2}\left(\begin{array}{ccccc} 0 & -1 & 1 & 0 & 0 \\ 0 & 1 & -1 & 0 & 0 \\0 & -1 & 1 & 0 & 0 \\0 & 0 & 0 & 0 & 0 \\0 & 0 & 0 & 0 & 0  \end{array}\right).$$
A direct calculation gives us $i S_1(x) f_\cup = i S_0(x) f_\triangle = i S_1(x) f_\triangle = 0$ and $i S_0(x) f_\cup = 1$.
Thus our approximation of $\frac{|W_{2n+1}(L_1 \triangle L_2)|}{|W_{2n+1}(L_1 \cup L_2)|}$ is $\frac00$.

A trim DFA for $L_1 \triangle L_2$ would have adjacency matrix 
$$\left(\begin{array}{ccc} 0 & 2 & 2 \\ 0 & 2 & 0 \\0 & 0 & 2  \end{array}\right)$$
with initial vector $(1,0,0)$, final vector $(0,1,1)^t$, and Rothblum's theorem gives us $\lambda = 2$, $q=1$, and 
$$S_0(x) = \left(\begin{array}{ccc} 0 & 1 & 1 \\ 0 & 1 & 0 \\0 & 0 & 1  \end{array}\right).$$
\end{example}

In the asymptotics of $J_C'(L_1, L_2)$, each congruence class $k$ is evaluated independently of the other congruence classes.
On the other hand, in $J_C(L_1, L_2)$ each congruence class $k$ has a limit that is a combination of $s_\ell'$ for all $0 \leq \ell < q$.
Thus the total answer is dominated by the overall asymptotic behavior and not just small periodic undercurrents.
This provides a rigorous explanation for Example \ref{J_C_example}.
It also allows for the following lemma to be the key to the convergence of $J_C$.

\begin{lemma}\label{one good one}
Using the notation of Theorem \ref{sum of matrices}, there exists $k$ such that $s_k' \neq 0$ when $A$ corresponds to a trim DFA.
\end{lemma}

\begin{proof}
Rothblum's theorem actually states that $\lim_{n \rightarrow \infty} \left(A/\lambda\right)^{qn+k} - S_k(qn+k) = 0$ with geometric convergence (see \cite{Rothblum_chapter}), i.e. there exists a matrix $C$ and number $\epsilon > 0$ such that $\left(A/\lambda\right)^{qn+k} - S_k(qn+k) < C(1-\epsilon)^{qn+k}$ for all $n$.
So if $S_k'(x) = 0$, then  $|W_{qn+k}(L)| \leq (\lambda(1-\epsilon))^{qn+k} i_A C f_A$.
But then $\limsup_{n \rightarrow \infty} \frac{\log|W_{qn+k}(L)|}{n} \leq \log((1-\epsilon)\lambda)$.
So by Theorem \ref{language_entropy}, there must be a $k'$ such that $S_{k'}'(x) \neq 0$.
\end{proof}

Recall that $|W_{qn+k}(L_1 \triangle L_2)| = \lambda^{qn+k} (S_k'(qn+k) + o(1))$, so it must be that $s_k'$ is a nonnegative real number for each $0 \leq k < q$.
Because $s_\ell' \geq 0$ for all $\ell$ and by Lemma \ref{one good one} there exists an $\ell'$ such that $s_{\ell'}' > 0$, it follows that $\frac1{q(d'+1)}\sum_{i=0}^{q-1}s_\ell',\frac1{\lambda^k(1-\lambda^{-q})}\sum_{\ell = k-q+1}^k \lambda^\ell s_\ell ' > 0$.
This concludes the proof to Theorem \ref{J_C_defined}.

\begin{remark}
In Example \ref{J_C_example} we stated that the our consideration of $J_C$ instead of $J_C'$ is by preference and not by necessity.
To emphasize this, we give a sketch of the fact that $J_C'$ is also well-defined (and thus our statements really are by choice and not by circumstance).
To do so, we will describe how to approximate more accurately $|W_{qn+k}(L)|$ for a fixed language $L$ and number $k$.

Define $L^{(q,k)}$ to be the set of words $w$ whose length is a multiple of $q$ and such that there exists a word $v$ of length exactly $k$ where $wv \in L$.
Suppose $L$ has adjacency matrix $A$, initial state vector $i$ and final states vector $f$.
The adjacency matrix for $L^{(q,k)}$ (prior to trimming) is $A^q$, the initial vector is the same ($i$), and the final states vector (with multiplicity) is $A^k f$.
So $|W_n (L^{(q,k)})| = i (A^q)^n (A^k f) = i A^{qn + k} f = |W_{qn+k} (L)|$.

It is well-known that if $A$ has period $q$, then $A^q$ is aperiodic (so $q=1$).
Moreover, it can be shown that (1) if $\mu$ is an eigenvalue of the trim subgraph of $A^q$, then $\mu = \lambda_*^p$ for some eigenvalue $\lambda_*$ of $A$, and (2) if $\lambda^q$ is an eigenvalue for the trim subgraph, then it is aperiodic.
Let $q'$ be the period of trimmed $L^{(q,k)}$.
If $L^{(q,k)}$ is aperiodic ($q' = 1$), then the growth of $|W_n (L^{(q,k)})|$ is defined by Rothblum's theorem with only one polynomial, and that polynomial is nonzero by Lemma \ref{one good one}, and so we are done.
If $q' > 1$, then iterate on $(L^{(q',0)})^{(q,k)} = L^{(qq',k)}$, which will be a subgraph with strictly fewer eigenvalues.

So if we applied this to Example \ref{J_C_example again} to calculate $|W_{2n+1}(L_1 \cup L_2)|$, the adjacency matrix would be 
$$A^2 = \left(\begin{array}{ccccc} 0 & 0 & 9 & 4 & 4 \\ 0 & 9 & 0 & 0 & 0 \\0 & 0 & 9 & 0 & 0 \\0 & 0 & 0 & 4 & 0 \\0 & 0 & 0 & 0 & 4  \end{array}\right),$$
and the final states vector would be $(4, 3, 0, 2, 2)^t$.
There are no longer paths from the initial state to the second state or from the third state to a final state.
Hence, the second and third states are removed when we trim the DFA.
After trimming, we see that 
$$ |W_{2n+1}(L_1 \cup L_2)| = |W_n (L^{(2,1)})| =  (1, 0 , 0) \left(\begin{array}{ccc} 0 & 4 & 4 \\ 0 & 4 & 0 \\0 & 0 & 4  \end{array}\right)^n \left(\begin{array}{c} 4 \\ 2 \\ 2  \end{array}\right) , $$
and the analysis follows easily from here.
\end{remark}


\subsection{Relationship between Entropy and Ces\'aro Jaccard}

We proved that the Ces\`{a}ro Jaccard distance is well-defined.  We have not proven that it is useful.
In fact, it is a rather rare instance that the Ces\`{a}ro Jaccard distance provides an interesting answer.

\begin{theorem} \label{limits of cesaro}
Let $L_1, L_2$ be two regular languages.
\begin{enumerate}
	\item If $h(L_1 \triangle L_2) \neq h(L_1 \cup L_2)$, then $J_C(L_1, L_2) = 0$.
	\item If $h(L_1 \cap L_2) \neq h(L_1 \cup L_2)$, then $J_C(L_1, L_2) = 1$. 
	\item If $0 < J_C(L_1, L_2) < 1$, then the following equal each other: \\
	$h(L_1),\; h(L_2),\; h(L_1 \cap L_2),\; h(L_1 \triangle L_2),\; h(L_1 \cup L_2)$.
	\end{enumerate}
\end{theorem}

\begin{proof}  
Part (1) easily follows from Lemma \ref{asymptotic growth}.
To see part (2), note that $(L_1 \cup L_2) \cap \overline{(L_1 \cap L_2)} = L_1 \triangle L_2$.
Therefore
\begin{eqnarray*}
\frac{|W_{\leq n}(L_1 \triangle L_2)|}{|W_{\leq n}(L_1 \cup L_2)|} &=& \frac{|W_{\leq n}(L_1 \cup L_2)| - |W_{\leq n}(L_1 \cap L_2)|}{|W_{\leq n}(L_1 \cup L_2)|}\\
								&=& \frac{2^{n(h(L_1 \cup L_2) - o(1))} - 2^{n(h(L_1 \cap L_2) - o(1))}}{2^{n(h(L_1 \cup L_2) - o(1))}}\\
								&=& 1 .
\end{eqnarray*}

For part (3), the above already implies that if $0 < J_C(L_1, L_2) < 1$, then $h(L_1 \cap L_2)$, $h(L_1 \triangle L_2)$, and $h(L_1 \cup L_2)$ are equal.
By symmetry, assume that $h(L_1) \leq h(L_2)$.
Because $L_1 \cap L_2 \subseteq L_1$ and $L_2 \subseteq L_1 \cup L_2$, by Lemma \ref{lemma:union_ent} we have that $h(L_1 \cap L_2) \leq h(L_1) \leq h(L_2) \leq h(L_1 \cup L_2)$.
Therefore all five terms are equal.
\end{proof}


\section{Entropy Distances}
\label{sec:entropydistance}

Entropy provides a natural method for dealing with the infinite
nature of regular languages. Because it is related to the eigenvalues of
the regular language's DFA, it is computable in polynomial time
given a DFA for the language. Note that the DFA does not have to be minimal.
We can therefore compute the entropy of set-theoretic combinations
of regular languages (intersection, disjoint union, etc) and use those
values to determine a well-founded distance between the languages.


\subsection{Entropy Distance}

A natural Jaccard-esque distance function based on entropy is the
entropy distance.
\begin{definition}[Entropy Distance]
Suppose $L_1$ and $L_2$ are regular languages.  Define the \textit{entropy distance} to be $H(L_1, L_2) = \frac{h(L_1\triangle L_2)}{h(L_1\cup L_2)}$.
If $h(L_1\cup L_2)$ is $0$, $H(L_1, L_2)=0$.
\end{definition}
This turns out
to be equivalent to a Jaccard limit with added $\log$ operations:

\begin{corollary}\label{log_Jaccard}
Suppose $L_1$ and $L_2$ are regular languages.  The following relation holds:
    \[\lim_{n\rightarrow\infty} \frac {\log \left|W_{\leq n}(L_1\triangle L_2)\right|}
                                      {\log \left|W_{\leq n}(L_1\cup L_2)\right|}
      = H(L_1, L_2).\]
\end{corollary}

\begin{proof}
Observe the following:

$$\lim_{n\rightarrow\infty} \frac {\log \left|W_{\leq n}(L_1\triangle L_2)\right|}{\log \left|W_{\leq n}(L_1\cup L_2)\right|}  = \lim_{n\rightarrow\infty} \frac {\frac{1}{n}\log \left|W_{\leq n}(L_1\triangle L_2)\right|}{\frac{1}{n}\log \left|W_{\leq n}(L_1\cup L_2)\right|} =
 \frac{h(L_1\triangle L_2)}{h(L_1\cup L_2)} = H_J(L_1,L_2).$$
Note that we can separate the limits because of Corollary \ref{language_entropy_cor}. 
\end{proof}

However, $H$ is not a good candidate for a distance function
as it only produces non-trivial results for languages that
have the same entropy.

\begin{proposition}
    \label{thm:entdis_neq}
    Suppose $L_1$ and $L_2$ are regular languages.
    If $h(L_1) \neq h(L_2)$, then $H(L_1, L_2) = 1$
\end{proposition}

\begin{proof}
    WLOG, suppose that $h(L_1)<h(L_2)$. First,
    $L_1\cap L_2\subseteq L_1$ which implies that $h(L_1\cap L_2)\leq h(L_1)$.
    Second,  $L_2\subseteq L_1\cup L_2$, and therefore $h(L_2)\leq h(L_1\cup L_2)$.
    All together this gives $h(L_1\cap L_2)<h(L_1\cup L_2)$, which implies that 
    $h(L_1\cap L_2)\neq h(L_1\cup L_2)$.
    By Lemma~\ref{lemma:union_ent}, $h(L_1\cup L_2) = \max(h(L_1\cap L_2), h(L_1\triangle L_2))$.  Thus, 
    $h(L_1\cup L_2)=h(L_1\triangle L_2)$.
   \end{proof}

As further evidence that $H$ is not a good candidate for
a distance function, we show it is an ultra-pseudo-metric. 
The ultra-metric condition, i.e. $d(x,z)\leq \max(d(x,y), d(y,z))$,
is so strong that it can make it difficult for the differences
encoded in the metric to be meaningful for practical applications.

\begin{theorem}\label{entropy_distance_metric}
    The function $H$ is an ultra-pseudo-metric.
\end{theorem}

\begin{proof}
The first two conditions of an ultra-pseudo-metric are satisfied by the definition of $H$ and from the
reflexiveness of $\triangle$ and $\cup$.  We now have to verify the ultra-metric inequality.

Suppose $L_1,L_2,L_3$ are regular languages.  We need to show that 
$$H(L_1,L_3)\leq \max (H(L_1,L_2),H(L_2,L_3)).$$

\noindent \textbf{Case 1:}  Suppose $h(L_1)\neq h(L_2)$.  By Proposition \ref{thm:entdis_neq}, $H(L_1,L_2)=1$.  Since $H(L_1,L_2)$ is a number between $0$ and $1$, $H(L_1,L_3)\leq \max (H(L_1,L_2),H(L_2,L_3))$.  

\noindent \textbf{Case 2:}  Suppose $h(L_1)=h(L_2)$.  If $h(L_3)\neq h(L_2)$, then the above argument holds.  Thus, assume that $h(L_1)=h(L_2)=h(L_3)$.  By Lemma \ref{lemma:union_ent}, $h(L_1\cup L_3)=h(L_1\cup L_2)=h(L_2\cup L_3)$.  Hence, it suffices to show that $h(L_1\triangle L_3)\leq \max (h(L_1\triangle L_2),h(L_2\triangle L_3))$.  Using multiple applications of Lemma~\ref{lemma:union_ent} we observe,
\begin{eqnarray*}
h\left(L_1\triangle L_3\right) &=& \max \left( h\left(L_1\cap \overline{L_3}\cap L_2\right), h\left(L_1\cap \overline{L_3}\cap \overline{L_2}\right), h\left(\overline{L_1}\cap L_3 \cap L_2\right), h\left(\overline{L_1}\cap L_3 \cap \overline{L_2}\right)\right)\\
&\leq& \max \left( h\left(L_1\cap \overline{L_2}\right), h\left(\overline{L_1}\cap L_2\right), h\left(L_2\cap \overline{L_3}\right), h\left(\overline{L_2}\cap L_3\right)\right)\\
&=& \max \left( h\left(L_1\triangle L_2\right), h\left( L_2\triangle L_3\right)\right).
\end{eqnarray*}

\end{proof}


\subsection{Entropy Sum}

In this subsection we will define a new (and natural) distance function for infinite regular languages.  We call this distance function the \textit{entropy sum distance}.  We will prove that not only is this distance function a pseudo-metric, it is also granular.  Granularity lends insight into the quality of a metric.  Intuitively, granularity means that for any two points in the space, you can find a point between them.  A metric $d$ on the space $X$ is \textit{granular} if for every two points $x,z\in X$, there exists $y\in X$ such that $d(x,y)<d(x,z)$ and $d(y,z)<d(x,z)$, i.e. $d(x,z)>\max (d(x,y), d(y,z))$.

\begin{definition}[Entropy Sum Distance]
    Suppose $L_1$ and $L_2$ are regular languages.  Define the \textit{entropy sum distance} to be 
    $H_S(L_1, L_2) = h(L_1\cap \overline{L_2}) + h(\overline{L_1}\cap L_2)$.
\end{definition}

The entropy sum distance was inspired by first considering the entropy of the symmetric difference directly, i.e. $h(L_1\triangle L_2)$.  However, since entropy measures the entropy of the most complex component (Theorem \ref{entropy_irr_components}), more information is gathered by using a sum as above in the definition of entropy sum.

\begin{theorem}\label{entropy_sum_metric}
    The function $H_S$ is a pseudo-metric.
\end{theorem}

\begin{proof}
The first two conditions of a pseudo-metric are satisfied by the definition of $H_S$ and from the
reflexiveness of $\triangle$ and $\cup$.  We now have to verify the triangle inequality.

Suppose $L_1,L_2,L_3$ are regular languages.  We need to show that 
$$H_S\left(L_1,L_3\right) = h\left(L_1\cap\overline{L_3}\right) + h\left(\overline{L_1}\cap L_3\right) \leq H_S\left(L_1,L_2\right)+H_S\left(L_2,L_3\right).$$
First observe,
\begin{eqnarray*}
h\left(L_1\cap\overline{L_3}\right) &=& \max\left(h\left( L_1\cap\overline{L_3}\cap L_2\right), h\left( L_1\cap\overline{L_3}\cap \overline{L_2}\right)\right)\\
&\leq& \max\left(h\left(L_2\cap \overline{L_3}\right), h\left(L_1\cap\overline{L_2}\right)\right).
\end{eqnarray*}
In a similar fashion, $h\left(\overline{L_1}\cap L_3\right)\leq h(\overline{L_1}\cap L_2)+h(\overline{L_2}\cap L_3)$.  Putting these together yields the desired result.
\end{proof}

The next two propositions display when granularity is achieved and when it is not.

\begin{proposition}\label{granular}
Let $L_1$ and $L_2$ be regular languages such that $h(L_1\cap \overline{L_2}), h(\overline{L_1}\cap L_2) > 0$.
Then, there exists two regular languages $R_1 \neq R_2$ such that $H_S(L_1,L_2) > \max (H_S(L_1,R_i), H_S(R_i,L_2))$ for each $i$.
\end{proposition}

\begin{proof}
Let $R_1 = L_1 \cup L_2$ and $R_2 = L_1 \cap L_2$. 
Notice that $H_S(L_1,R_1) = h(\overline{L_1}\cap L_2)$ and $H_S(R_1,L_2)=h(L_1\cap\overline{L_2})$.  Hence, 
$$H_S(L_1,L_2)=h(L_1\cap\overline{L_2})+h(\overline{L_1}\cap L_2) > \max(H_S(L_1,R_1), H_S(R_1,L_2)).$$
The statement involving $R_2$ is analogous.
\end{proof}

\begin{proposition}\label{not_granular}
Let $L_1$ and $L_2$ be regular languages such that $h(\overline{L_1}\cap L_2) = 0$.
For all regular languages $L$ we have that $H_S(L_1,L_2) \leq \max (H_S(L_1,L), H_S(L,L_2))$.
\end{proposition}

\begin{proof}

Note that $H_S(L_1,L_2) = h(L_1\cap \overline{L_2})$.
The proof breaks down into two cases:

\noindent \textbf{Case 1:} Suppose $h(L_1\cap\overline{L_2}\cap L)=h(L_1\cap\overline{L_2})$.  Then, 
\begin{eqnarray*}
h\left(L_1\cap \overline{L_2}\right) &=& h\left(L_1\cap\overline{L_2}\cap L\right) \\
&\leq& h\left(L\cap\overline{L_2}\right)\\
&\leq& H_S\left( L,L_2\right)\\
&\leq& \max \left(H_S\left(L_1,L\right), H_S\left(L,L_2\right)\right).
\end{eqnarray*}

\noindent \textbf{Case 2:} Suppose $h(L_1\cap\overline{L_2}\cap L) <  h(L_1\cap\overline{L_2})$.  Then, 
\begin{eqnarray*}
h\left(L_1\cap \overline{L_2}\right) &=& \max\left( h\left(L_1\cap \overline{L_2}\cap L\right), h\left(L_1\cap \overline{L_2}\cap \overline{L}\right)\right)\\
&=& h\left(L_1\cap \overline{L_2}\cap \overline{L}\right)\\
&\leq& h\left( L_1\cap\overline{L}\right)\\
&\leq& H_S\left( L_1,L\right)\\
&\leq& \max \left(H_S\left(L_1,L\right), H_S\left(L,L_2\right)\right).
\end{eqnarray*}
\end{proof}


\section{Conclusion and Future Work}
\label{sec:conclusion}

This paper has covered some issues related to the entropy of and distance
between regular languages. It has proven correct the common upper limit 
formulation of language entropy and has provided a limit based entropy
formula that can be shown to exist. Jaccard distance was shown to be related
to language entropy, and various limit based extensions of the Jaccard
distance were shown to exist or not exist. The natural entropy based
distance function was shown to be an ultra-pseudo-metric, and some facts
were proven about the function that show it likely to be impractical.
Finally, the paper introduces an entropy-based distance function and proves
that function to be a pseudo-metric, as well as granular under certain conditions. 

In this paper several formulations of entropy are developed, and it is natural
to consider which would be the best to use. In a practical sense it does not
matter since all formulations are equivalent (Theorem~\ref{language_entropy}) and 
can be computed using Shannon's determinant-based method
\cite{shannon48mathematical}. However, conceptually, it can be argued that
$\lim_{n\rightarrow\infty}\frac {log \left|W_{\leq n}(L)\right|} {n}$ is the preferable
formulation. First, there is a notational argument that prefers using limits
that exist. This is a limit that exists (Corollary~\ref{language_entropy_cor}), whereas many other limit formulations do
not. Second, this limit captures more readily the concept of ``number of bits
per symbol'' that Shannon intended.  Because regular languages can have strings
with staggered lengths, using $W_n$ forces the consideration of possibly empty
sets of strings of a given length. This creates dissonance when the language
has non-zero entropy. Instead, the monotonically growing $W_{\leq n}$ more
clearly encodes the intuition that the formulation is expressing the number of
bits needed to express the next symbol among all words in the language.

Apart from expanding to consider context-free languages and other languages (\cite{Cui_Dang_etal}), one investigation
that is absent from this paper is the determination of similarity between
languages that are disjoint but obviously similar (i.e.  $aa^*$ and $ba^*$).
A framework for addressing such problems is provided in
\cite{cui2013similarity}, but finding metrics capturing
such similarities can be fodder for future efforts.

\bibliography{refs}

\newpage
\appendix

\end{document}